\RequirePackage{fix-cm}
\documentclass[smallextended]{svjour3}       
\smartqed  
\usepackage{graphicx}
\usepackage[utf8]{inputenc}
\usepackage[T1]{fontenc}
\usepackage{amsmath,amsfonts,amssymb}
\usepackage{algorithm}
\usepackage[noend]{algpseudocode}
\usepackage{graphicx}
\usepackage[labelfont=bf, labelsep=space]{caption}
\usepackage{subcaption}
\usepackage{natbib}
\usepackage{array}
\usepackage{hyperref}
\usepackage{url}

\begin{document}

\title{Randomized rounding algorithms for large scale unsplittable flow problems
}

\titlerunning{Randomized rounding algorithms}        

\author{François Lamothe \and
Emmanuel Rachelson \and
Alain Haït \and
Cedric Baudoin \and
Jean-Baptiste Dupe
}


\institute{François Lamothe \at
              ISAE-SUPAERO, Université de Toulouse, France \\
              Corresponding author \\
              ORCID : 0000-0002-1332-9206 \\
              \email{francois.lamothe@isae-supaero.fr}
           \and
           Alain Haït \at
              ISAE-SUPAERO, Université de Toulouse, France \\
              ORCID : 0000-0001-5918-0753
           \and
           Emmanuel Rachelson \at
              ISAE-SUPAERO, Université de Toulouse, France \\
              ORCID : 0000-0002-8559-1617
           \and
           Cedric Baudoin \at
              Thalès Alenia Space, Toulouse, France \\
           \and
           Jean-Baptiste Dupe \at
              Centre national d'études spatiales (CNES), Toulouse, France
}

\date{Received: date / Accepted: date}

\maketitle

\begin{abstract}
Unsplittable flow problems cover a wide range of telecommunication and transportation problems and their efficient resolution is key to a number of applications. In this work, we study algorithms that can scale up to large graphs and important numbers of commodities. We present and analyze in detail a heuristic based on the linear relaxation of the problem and randomized rounding. We provide empirical evidence that this approach is competitive with state-of-the-art resolution methods either by its scaling performance or by the quality of its solutions. We provide a variation of the heuristic which has the same approximation factor as the state-of-the-art approximation algorithm. We also derive a tighter analysis for the approximation factor of both the variation and the state-of-the-art algorithm. We introduce a new objective function for the unsplittable flow problem and discuss its differences with the classical congestion objective function. Finally, we discuss the gap in practical performance and theoretical guarantees between all the aforementioned algorithms.
\end{abstract}

\section*{Declarations}

This work was partially funded by Thales Alenia Space and made in collaboration with several of its members. This work was partially funded by the CNES. Several authors are academically related to ISAE-SUPAERO.

The authors declare that they have no conflict of interest.
The datasets and the code used in the experimental section of this work are accessible at \href{https://github.com/SuReLI/randomized_rounding_paper_code}{https://github.com/SuReLI/randomized\_rounding\_paper\_code}.

\section{Introduction}

The unsplittable flow problem is an extensively studied variant of the classical maximum flow problem. In this problem, one is given a directed or undirected graph, together with capacities on its arcs. A family of commodities, each composed of an origin, a destination, and a demand, is also provided. Each commodity has to route its demand from its origin to its destination through a unique path. The routing must ensure that capacities on the arcs are not exceeded by the flow of the commodities, or at least minimize the violation of the capacities.

This problem is NP-hard as it contains several NP-hard problems as sub-cases. When there are only two nodes linked by one arc, the problem is equivalent to the knapsack problem. When all demands and capacities are 1, the problem is equivalent to the edge-disjoint paths problem.

This problem has various applications especially in telecommunication networks (\emph{e.g.} optical networks, telecommunication satellites), and goods transportation. In these applications, large-scale instances appear with up to 500 nodes, 2000 arcs, and 150 000 commodities. However, only a few methods in the literature can scale to such large instances, such as the approximation algorithm of \citet{raghavan1987randomized} and some meta-heuristics tuned to have small computing times. The algorithm presented by \citet{raghavan1987randomized} uses randomized rounding to compute a solution to the unsplittable flow problem. Even though this $O\left(\frac{\ln m}{\ln \ln m}\right)$-approximation algorithm has theoretically the best approximation factor achievable, the solution it yields are often far from optimal in terms of solution quality.

That is why, in this work, we focus on an algorithm that can scale to large instances while giving good practical results. This algorithm is a heuristic extension of the randomized rounding algorithm of \citet{raghavan1987randomized}. As such, it also uses randomized rounding on the linear relaxation of the unsplittable problem to create an unsplittable solution. This algorithm alternates randomized rounding steps and resolutions of the linear relaxation and will thus be called, in this work, the Sequential Randomized Rounding algorithm (SRR). This heuristic is also an extension of the algorithm proposed by \citet{rivano2002lightpath} for which no complete proof of the approximation factor was given. Compared to the algorithm of \citet{rivano2002lightpath}, the SRR heuristic offers more flexibility on the number of linear relaxation resolutions and more importantly, takes advantage of the fact that commodities might have different demand levels. We also describe a variation of the SRR heuristic for which we prove the same approximation guarantees as the algorithm of \citet{raghavan1987randomized}. This variation will be called the Constrained Sequential Randomized Rounding algorithm (CSRR). Moreover, we tighten the analysis of both approximation algorithms and prove that they achieve a $O\left(\frac{\gamma \ln  m}{\ln(\gamma \ln m)}\right)$-approximation factor where $\gamma$ is a parameter that is small when the commodities demands are small compared to the capacities of the arcs. Finally, we experimentally show that the SRR algorithm scales to large instances. Furthermore, its practical results on large datasets outperform other tested methods.

This paper is structured as follows. In Section \ref{sec:problem} we describe the notations used together with several Mixed Integer Linear Programs (MILP) for the unsplittable flow problem. Related work is presented in Section \ref{sec:related}. Section \ref{sec:algo} presents the SRR heuristic and its complexity analysis. Section \ref{sec:proof} describes the CSRR algorithm and provides the analysis leading to the $O\left(\frac{\gamma \ln m}{\ln(\gamma \ln m)} \right) $-approximation factor. In Section \ref{sec:results}, we provide experimental results that compare the empirical quality of the various algorithms presented. In Section \ref{sec:discussion}, we discuss different properties of the SRR heuristic and the CSRR algorithm. Finally, we conclude and give perspectives in Section \ref{sec:conclusion}.

\section{The unsplittable flow problem}
\label{sec:problem}

Throughout this paper, we will use the following notations:
\begin{itemize}
    \item $G = (V,E)$ is a directed or undirected graph, with $V$ the set of nodes and $E$ the set of arcs
    \item $L = (o_k, d_k, D_k)_{k\in K}$ is a set of commodities defined by their origin, destination and demand.
    \item $(c_e)_{e \in E}$ are capacities on the arcs
\end{itemize}

We also use the Kronecker notation, $\delta^y_x$ equals 1 if $x = y$ and 0 otherwise. The sets of arcs incoming and outgoing of node $v$ will be noted $E^-(v)$ and $E^+(v)$ respectively. Finally, $c_{min} = \min_{e \in E} c_e$ and $D_{max} = \max_{k \in K} D_k$ are the smallest capacity and the largest demand.

\subsection{Objective functions}

\label{two_criteria}

Four objective functions for the unsplittable flow problems can be found in the literature: maximizing the served demand \citep{kolman2003note}, minimizing the cost \citep{barnhart2000using}, minimizing the congestion \citep{martens2006flows}, minimizing the number of necessary routing rounds \citep{aumann1995improved}. In this work, we focus on minimizing the congestion. 

The congestion of a graph is the smallest number $\Delta$ by which it is necessary to multiply all the capacities in order to route all the commodities \citep{martens2006flows}. The congestion of an arc is the ratio of the flow on this arc to its capacity. This metric emphasizes low capacity arcs. Besides, minimizing the congestion puts no restrictions on the flow going through the arcs that do not have a maximal congestion. In particular, it induces no incentive to minimize the congestion on those arcs. This becomes problematic when an arc is largely more congested in every solution than any other arc because it lifts all restrictions for the other arcs. To prevent this, we introduce a new criterion to minimize the violation of the capacities of the arcs. We use the term \emph{overflow} to describe the quantity of flow $\Delta_e$ that exceeds the capacity of an arc $e$. The overflow of an arc is always non-negative. Our new criterion is to minimize the sum of the overflows $\sum_{e \in E} \Delta_e$. Note that the congestion $\Delta$ is not the maximum overflow over all the arcs.

In the following sections, we present two equivalent Mixed Integer Linear Programs (MILP) for the unsplittable flow problem.

\subsection{Arc-node formulation}
\label{sec:arc-node-formulation}

The arc-node formulation is compact as it has a polynomial number of variables and constraints in the number of commodities, nodes, and arcs. It can thus be solved directly in a MILP solver for small instances. This formulation is characterized by the flow conservation constraints which ensure the structural property of flows. The objective function presented is the overflow sum.
The meaning of the variables in this formulation is the following:
\begin{itemize}
    \item $f_{ek}$ indicates whether commodity $k$ pushes flow on arc $e$,
    \item $\Delta_e$ represents the overflow on arc $e$.
\end{itemize}

The unsplittable flow problem is then:

\begin{subequations}
\begin{alignat}{2}
&\min_{f_{ek},\Delta_e}  \sum_{e \in E} \Delta_e & \qquad &  \\
&\text{such that} &      & \notag\\
&\sum_{e \in E^+(v)} f_{ek} ~  - \sum_{e \in E^-(v)} f_{ek} = \delta^{o_k}_v - \delta^{d_k}_v & & \forall k \in K, ~ \forall v \in V, \label{eq: flow_const_na}\\
&\sum_{k \in K} f_{ek}  D_k \leq c_e + \Delta_e & & \forall e \in E,  \label{eq: capacity_const_na}\\
& f_{ek} \in \{0,1\}, ~ \Delta_e \in \mathbb{R}^+ & & \forall k \in K, ~ \forall e \in E.
\end{alignat}
\end{subequations}

Equation (\ref{eq: flow_const_na}) corresponds to the flow conservation constraints. It ensures that, for each commodity and every node except the origin and the destination of the commodity, the same amount of flow of the commodity goes in and out of the node. Equation (\ref{eq: capacity_const_na}) corresponds to the capacity constraints. It ensures that the capacity of an arc is respected or that the overflow is recorded in $\Delta_e$. The fact that $ f_{ek} \in \{0,1\}$ ensures that the flow is unsplittable.

One can create an arc-node congestion formulation by replacing $c_e + \Delta_e$ by $c_e \Delta$ in Equation \eqref{eq: capacity_const_na} and minimizing over $\Delta$ instead of $\sum_e \Delta_e$. Note that the $\Delta$ variable is common to all arcs while there was one variable $\Delta_e$ for each arc.

\subsection{Path formulation}
\label{sec:path_formulation}

In the path formulation, the flow conservation constraints are unnecessary. The variables represent paths so these constraints are always verified. However, there is an exponential number of variables (in the number of commodities, nodes, and arcs) so the formulation must be solved through a particular MILP technique named Branch and Price \citep{barnhart2000using}. The objective function presented is the overflow sum. The meaning of the variables in this formulation is the following:
\begin{itemize}
    \item $x_{pk}$ indicates whether commodity $k$ uses path $p$ to push its flow,
    \item $\Delta_e$ represents the overflow on arc $e$.
\end{itemize}

The unsplittable flow problem is then:

\begin{subequations}
\begin{alignat}{2}
&\min_{x_{pk},\Delta_e}  \sum_{e \in E} \Delta_e & \qquad &  \\
&\text{such that} &      & \notag\\
&\sum_{p \in P_k} x_{pk} = 1 & & \forall k \in K, \label{eq: convex_p}\\
&\sum_{k \in K} \sum_{p \in P_k | e \in p} x_{pk} D_k \leq c_e + \Delta_e & & \forall e \in E,  \label{eq: capacity_const_p}\\
& x_{pk} \in \{0,1\}, ~ \Delta_e \in \mathbb{R}^+ & & \forall p \in \bigcup_k P_k, ~ \forall k \in K, ~ \forall e \in E.
\end{alignat}
\end{subequations}

Here, $P_k$ denotes the set of all paths usable by commodity $k$. Equation (\ref{eq: convex_p}) ensures that exactly one path is chosen for each commodity. Equation (\ref{eq: capacity_const_p}) corresponds to the capacity constraints. It ensures that the capacity of an arc is respected or that the overflow is recorded in $\Delta_e$. The fact that $ x_{pk} \in \{0,1\}$ ensures that the flow is unsplittable.

\section{Related work}
\label{sec:related}

In this section, we review important solution approaches to the unsplittable flow problem present in the literature. These works are grouped into three sub-sections: exact methods, approximation algorithms, and meta-heuristics. A fourth sub-section is dedicated to the linear relaxation of the unsplittable flow problem (the multi-commodity flow problem) whose resolution is paramount to the resolution of the unsplittable flow problem.

\subsection{Exact methods}

\citet{barnhart2000using} presented a Branch and Price and Cut procedure applied to a path formulation with a cost minimization objective function. Most subsequent works use this baseline as a comparison. A major contribution of their work is their branching strategy.  Unlike straightforward branching strategies for this problem, theirs keeps intact the structure of the pricing problem throughout the branching procedure. For a commodity in a non-integer solution, the divergence node is the first node where the flow of the commodity splits. The outgoing arcs of the divergence node are divided into two disjoint subsets $E_1$ and $E_2$. Each set must contain at least one of the arcs used by the commodity. The branching rule is: either forbid the use of $E_1$ or forbid the use of $E_2$. In both cases, the previous non-integer solution is cut from the problem and forbidding sets of arcs keeps the structure of the pricing problem intact. They also included cuts to strengthen the linear relaxation. These cuts are lifted cover inequalities of the capacity constraints.

\citet{park2003integer} mixed the path formulation and a knapsack formulation (which is not presented in this work) to derive a new linear formulation of the problem. The linear relaxation of this formulation yields a stronger lower bound, which in turn decreases the time needed to complete the branching procedure. They compared different branching rules and report that the one of \citet{barnhart2000using} produces much better results. They present computational results only for this rule.

\citet{belaidouni2007minimum} presented a cutting plane method based on super-additive functions to get strong cuts for their Branch and Price method. It appears on small instances that the addition of their cuts derives integer solutions without using a Branch and Bound procedure. Results are compared with those of \citet{barnhart2000using} and large improvements are reported.

Overall, the best results can be found in the articles of \citet{belaidouni2007minimum} and \citet{park2003integer}. \citet{belaidouni2007minimum} compared their results with those of \citet{barnhart2000using} and solved all their instances (at most 30 nodes, 60 arcs, 100 commodities) in less than 10 seconds. \citet{park2003integer} did not compare their results with previous works but solved instances of the same magnitude (at most 30 nodes, 80 arcs, 100 commodities) in less than 15 seconds. Note that results were only given for small instances in these approaches.

Other earlier works have been reported in \citep{parker1993column, alvelos2003comparing, park1996integer}.

\subsection{Approximation algorithms and heuristics}
\label{sec:approx}

As the unsplittable flow problem is NP-hard, a lot of attention has been given to approximation algorithms and heuristics. In particular, the maximum served demand objective has been extensively studied. We refer to the Handbook of approximation algorithms \citep{taylor2020handbook} for a detailed survey on approximation algorithms in the context of unsplittable flows. We recall here some works related to the minimum congestion objective.

Approximation algorithms are given with a factor of approximation $\lambda$. Let $\Gamma$ be the objective function of the minimization problem at hand. Solutions generated by a $\lambda$-approximation algorithm verify the following inequality:
$$\Gamma(S^*) \leq \Gamma(S) \leq \lambda \Gamma(S^*),$$
where $\Gamma(S)$ and $\Gamma(S^*)$ are respectively the value of the produced solution and the value of the optimal solution.
This guarantees that the ratio between the value of the produced solution and the value of the optimal solution is not too high. When considering approximation algorithms, more attention must be paid to the objective function being optimized. The literature uses four different objective functions: maximizing the served demand, minimizing the congestion, minimizing the number of rounds, and minimizing the cost. Approximation hardness results demonstrate that none of these objective functions admits a constant-factor approximation algorithm \citep{taylor2020handbook}.

In the congestion minimization context, we seek the smallest number by which it is necessary to multiply all the capacities in order to fit all the commodities. The best-known approximation algorithm for congestion is a randomized rounding method introduced by \citet{raghavan1987randomized} which we shall call the RR algorithm in this work. The method proceeds in two steps. First, a solution of the linear relaxation of the problem is computed. Each commodity is allowed to use multiple paths in this solution. The proportion of flow for each commodity on each path is $((x_{pk})_{p \in P_k})_{k \in K}$. Then a path is selected for each commodity. Path $p$ is selected in $P_k$ with probability $x_{pk}$. Each commodity is then assigned to the selected path to create an unsplittable solution. This procedure produces, with arbitrarily high probability, an unsplittable solution whose congestion is $O\left(\frac{\ln |E|}{\ln \ln |E|}\right)$ larger than the one of the fractional solution. Their algorithm is thus a $O\left(\frac{\ln |E|}{\ln \ln |E|}\right)$-approximation algorithm that works for directed and undirected graphs. The randomized rounding process can be derandomized using the method of conditional probabilities \citep{raghavan1988probabilistic}. \citet{chuzhoy2007hardness} showed a tight $\Omega(\frac{\ln |V|}{\ln \ln |V|})$ bound on directed graphs, assuming $NP \nsubseteq BPTIME(|V|^{O(\ln \ln |V|)})$. \citet{andrews2010inapproximability} showed that minimizing congestion for an unsplittable flow in an undirected graph is hard to approximate within  $\Omega(\ln \ln |V|/\ln \ln \ln |V|)$, assuming $NP \nsubseteq ZPTIME(|V|^{polylog(|V|)})$. A $(|K|+2)$-approximation algorithm is presented in \citep{asano2000experimental}. The reported results show that in practice this algorithm gives results comparable to classical randomized rounding.

In the context of the maximum demand objective function, a parameter $\frac{D_{max}}{c_{min}}$ is introduced. This parameter plays an important role for this objective function. Indeed, when this parameter is upper-bounded by a small constant, several works reported stronger approximation results for their algorithms \citep{chakrabarti2007approximation, shepherd2015inapproximability, azar2006combinatorial}. However, we did not find similar results for the congestion objective function.

A few heuristics have been proposed in previous works. They are either greedy or Linear programming (LP) based heuristics. \citet{rivano2002lightpath} introduced an algorithm very similar to the SRR algorithm presented in Section \ref{sec:algo}, without proving it is an approximation algorithm. \citet{asano2000experimental} as well as \citet{wang1999explicit} proposed greedy algorithms and LP-based algorithms. Reported results show that, except in specific cases, the greedy approaches are usually not competitive with LP-based methods. On the other hand, LP-based heuristics yield results that are similar to the randomized rounding algorithm of \citet{raghavan1987randomized}.

\subsection{Meta-heuristics}

As introduced above, it is NP-hard to find an optimal solution or even to give a constant-factor approximation to the unsplittable flow problem. Thus, the literature investigated various randomized search procedures such as genetic algorithms \citep{cox1991dynamic, masri2019metaheuristics}, tabu search \citep{anderson1993path, laguna1993bandwidth, xu1997tabu}, local search and GRASP \citep{santos2013optimizing, santos2013hybrid, santos2010link, alvelos2007local, masri2015multi, masri2019metaheuristics} or ant colony optimization  \citep{li2010ant, masri2011ant}. One of the major difficulties encountered when solving the unsplittable flow problem with a meta-heuristic is to efficiently create useful paths for the commodities. 

Early approaches such as \citep{cox1991dynamic, anderson1993path} encode solutions as permutations of the commodities. The space of permutations is the one searched by the meta-heuristic. The following function is used to create a solution from a permutation and evaluate it. The function goes through the permutation, examining each commodity. The commodity is then allocated to the shortest path where there is still enough capacity to fit the commodity. Once a path is assigned to every commodity, the objective function can be computed.

In \citep{laguna1993bandwidth} and \citep{masri2015multi} the $k$ shortest paths are pre-computed for each commodity using the algorithm of \citet{yen1971finding}. The search space of their meta-heuristics is restrained to the space of solutions using only those paths.

A different idea used in \citep{santos2013optimizing, santos2013hybrid, santos2010link, alvelos2007local} is to consider paths extracted from the linear relaxation of the problem. The linear relaxation is solved with a column generation algorithm applied to the path formulation. During the column generation, a set of paths $\hat{P}_k$ is generated for each commodity. A meta-heuristic such as a multi-start local search is then used to explore the solutions where only paths from $\hat{P}_k$ are used. In \citep{santos2013optimizing}, after the first linear relaxation is solved, perturbed linear models are solved to create new useful columns and extend the solution space of the meta-heuristic.

Ant colony optimization is also a means to navigate the large solution space of the possible paths and is used in \citep{li2010ant, masri2011ant}. In an ant colony optimization approach, at each iteration, each commodity creates a path by taking into account several metrics: path length, path load, and pheromones. Each arc of the graph has a pheromone level for each commodity and the higher the pheromone level the higher the probability of the arc to belong to the path generated. Pheromones are updated through two means. First, the best solutions add pheromones to the arc they use. Second, pheromones decay so that their level does not become excessive thus facilitating the exploration of the solution space.

We refer to the work of \citet{li2010ant} and \citet{santos2013hybrid} for the best performing meta-heuristics. \citet{li2010ant} compared their results with the solver CPLEX and were able to solve instances with up to 60 nodes, 400 arcs, and 3500 commodities to optimality in less than 900 seconds. \citet{santos2013hybrid} show that all their instances (26 nodes, 80 arcs, 500 commodities) are solved in less than 180 seconds with values close to the linear relaxation lower bound.

\subsection{Linear multi-commodity flow problem} 
\label{LMCFP}

The multi-commodity flow problem is the linear relaxation of the unsplittable flow problem. The value of its optimal solution is a lower bound for the binary problem and this linear relaxation is used in several exact and approximate methods. As a special case of linear programming, this problem is solvable in polynomial time.

A lot of effort has been invested in exact methods for this problem. Even if a commercial solver can solve the node-arc formulation, this method may take a prohibitive time in large instances. An alternative solution is to use a Lagrangian relaxation of the capacity constraints to decompose the problem into easier sub-problems as in \citep{retvdri2004novel}. As reported by \citet{weibin2017solving}, Lagrangian relaxation shows lesser performances in most instances than the competitor method, applying a column generation algorithm to the path formulation. Lagrangian relaxation seems to be the best choice when the number of commodities is very large because its computing time scales only linearly with this parameter. In most other cases, column generation seems to be the best solution. 

Several works contributed to increase the performance of column generation algorithms for this problem. First, a primal-dual column generation is presented in \citep{gondzio2013new, gondzio2015new, gondzio2016large}. An interior-point algorithm is used to solve the master problem and obtain sub-optimal but well-centered solutions. These well-centered solutions are used to compute new columns in the sub-problems which stabilizes the column generation process and reduces the number of iterations needed to achieve convergence. Another approach, which could be combined with the previous one is the use of aggregated variables presented by \citet{bauguion2013new, bauguion2015efficient}. In this method, variables do not represent paths but aggregated paths such as trees or more complex structures. The authors report that the sub-problems associated with aggregated variables can be solved efficiently. Aggregated variables reduce the size of the master problem during the algorithm but might induce a larger number of iterations, thus aggregation must be carefully done. Another method presented in \citep{babonneau2006solving} consists of a specialized interior-point method to solve the multi-commodity flow problem. This method has been improved by \citet{castro2012improving}. Other contributions to linear programming methods can be found in \citep{moradi2015bi, dai2016finding, dai2016node}

For large instances, linear programming methods may take a lot of computing time before finding the optimal solution. Thus the literature focused on combinatorial approximation algorithms and in particular on Fully Polynomial-Time Approximation Schemes (FPTAS). The best results were obtained through the use of exponential length functions. This idea was first introduced by \citet{shahrokhi1990maximum}. In their algorithm, a length exponential in the passing flow is assigned to each arc. The flow is iteratively augmented on the shortest path connecting any of the source-sink pairs. Their algorithm was improved by \citet{fleischer2000approximating} who showed that only the computation of an $\epsilon$-shortest path was needed. The algorithm of \citet{fleischer2000approximating} is the fastest FPTAS in practice while not being the one with the smallest complexity. Indeed, \citet{madry2010faster} presented an algorithm with a smaller complexity but \citet{emanuelsson2016approximating} showed in his work that the algorithm of \citet{fleischer2000approximating} is faster for instances having less than 100 million arcs. For a detailed survey on the multi-commodity flow problem methods previous to 2005, see the article of \citet{wang2018multicommodity}.

\section{The Sequential Randomized Rounding heuristic}
\label{sec:algo}

The Sequential Randomized Rounding algorithm (SRR) is a polynomial greedy heuristic for the unsplittable flow problem. This algorithm is similar to the one proposed by \citet{rivano2002lightpath} for the light-path assignment problem. We add some features that are specific to our problem such as the sorting of the commodities by decreasing demand. The SRR algorithm also gives the possibility to compute less often the linear relaxation of the problem to reduce the overall computing time of the algorithm. Compared to the approximation algorithm of \citet{raghavan1988probabilistic}, the SRR algorithm has larger running times, no performance guarantees but returns higher quality solutions on the tested instances. As a heuristic, the SRR algorithm has a shorter computing time than exact solutions and meta-heuristics, especially for large instances. A variation of the SRR heuristic with approximation guarantees is presented in Section \ref{sec:proof} together with an analysis of these approximation guarantees. Finally, the SRR algorithm is further discussed in Section \ref{sec:discussion} to explain its behavior on the tested instances.

\subsection{Presentation of the algorithm}

The SRR algorithm, presented in Algorithm \ref{heuristic} alternates between two different steps: solve the linear relaxation of the problem and fix some commodities to a unique path. In our case, the linear relaxation is a multi-commodity flow problem. Even though for the sake of clarity we use the notations of the path formulation in the following, in the experimentations, an arc-node formulation paired with the commercial solver \citep{gurobi} is used to solve the linear relaxation. More efficient specialized solvers or approximation algorithms for the multi-commodity flow problem can be found in the literature (see Section \ref{LMCFP}). Solving the linear relaxation provides a distribution of flow among the paths for each commodity: $((x_{pk})_{p \in P_k})_{k \in K}$. After solving the linear relaxation, a path is selected for some commodities. These commodities will be forced to use only these paths for the rest of the algorithm. The path selected for each commodity is chosen through the same randomized rounding procedure introduced by \citet{raghavan1987randomized}: for commodity $k$, path $p$ is selected with probability $x_{pk}$. The probability that commodity $k$ uses arc $e$ is $f_{ek} = \sum_{p \in P_k | e \in p} x_{pk}$. When solving the next linear relaxations, the fixed commodities will also be forced to only use their single allowed path. 

The major difference with the RR algorithm of \citet{raghavan1987randomized} is that, in the SRR heuristic, the linear relaxation is actualized several times during the randomized rounding process. More precisely, after deciding to fix some commodities to a single path, the linear relaxation is solved again with the added constraints that the fixed commodities must use their affected path. To decide when the linear relaxation is actualized, the following procedure is used. In the solution of the linear relaxation, some commodities use multiple paths. After $\theta$ of these commodities are fixed to a single path, the linear relaxation is actualized. Choosing the threshold $\theta$ trades off between computation time and solution quality. If the linear relaxation is actualized often (low threshold), fixing decisions take into account most of the previous fixing decisions but the computation time is high. If the linear relaxation is never actualized, branching decisions do not take into account previous branching decisions but the computation time is low. The threshold value is fixed to $\frac{|V|}{4}$ in our experiments. A sensitivity analysis of this parameter is given in Section \ref{sec:results}.

Another difference with the RR algorithm is that solutions are created using of the overflow sum objective function presented in Section \ref{two_criteria} instead of the classical congestion objective function. As will be explained in Section \ref{compare_criteria}, when using the overflow sum objective function, the SRR heuristic returns solutions with a lower overflow but also a lower congestion.

Compared to the algorithm proposed in \citep{rivano2002lightpath}, the SRR heuristic offers more flexibility on the number of actualizations of the linear relaxation thanks to parameter $\theta$. Moreover, the main difference is that, because the unsplittable flow problem in \citep{rivano2002lightpath} arises from a Light-Path Assignment problem, all commodities have a demand of one and thus the commodities have their path chosen in no particular order. In the SRR heuristic, paths are assigned to the commodities in decreasing order of commodity's demand, \emph{i.e.} the commodities with larger demands have their paths chosen first. This order is classically used in bin packing heuristics such as the Next Fit Decreasing heuristic \citep{csirik1986probabilistic}. The rationale behind this ordering is to allocate commodities with a large demand first, while a large amount of capacity is left in the arcs. Commodities with smaller demands are then used to fill the remaining gaps. In Section \ref{sec:results}, it is shown that this rounding order of the variables has a large impact on the quality of the solution returned by the heuristic.

\begin{algorithm}
    \caption{The SRR heuristic}
    \begin{algorithmic}[1]
    \Require $G = (V,E,c)$ a capacitated graph, $L = (o_k, d_k, D_k)_{k \in K}$ a list of commodities
    \State Sort the commodities by decreasing demand \label{algo:line1}
    \State $K_{fixed} = \varnothing$ \Comment{$K_{fixed}$ is the set of indices of commodities fixed to a single path}
    \For{each commodity $k^*$ in decreasing demand order}
    \If {an actualization is needed}
    \State $((x_{pk})_{p \in P_k})_{k \in K} = Solve\_Linear\_Relaxation \left( G, L, K_{fixed}, (p_k)_{k \in K_{fixed}} \right)$ \label{algo:lineLP}
    \EndIf
    \State Draw a path $p^*$ from $P_{k^*}$ with probability $x_{p^*k^*}$ \label{algo:line5}
    \State Add index $k^*$ to $K_{fixed}$. \label{algo:remove}
    \State $p_{k^*} = p^*$ \label{algo:linefixing}
    \EndFor
    \State \Return $(p_k)_{k \in K}$
    \end{algorithmic}
    \label{heuristic}
\end{algorithm}

\subsection{Complexity analysis}
There are $|K|$ iterations and during each of them the algorithm might do the following:

\begin{itemize}
    \item Solve the linear relaxation: $O(LR(|V|, |E|, |K|))$ operations where $LR(|V|, \\ |E|, |K|)$ is the complexity of the linear relaxation resolution (line \ref{algo:lineLP}).
    \item Select $p_*$: as the flow of each commodity can be decomposed in at most $|E|$ paths, at most $|E|$ of the variables $x_{pk}$ have a non-zero value \citep{ford1956network}. Choosing one of these variables requires $O(|E|)$ operations (line \ref{algo:line5}).
\end{itemize}

Additionally, sorting the commodities requires $O(|K| \log(|K|))$ operations (line \ref{algo:line1}). The total time complexity is $O(|K|(\log |K| + |E| + LR(|V|, |E|, |K|)))$.

\subsection{Grouping commodities by origin}
\label{sec:aggregation}

Grouping commodities by origin is the process of considering a set of commodities with the same origin as one single commodity. This new commodity has a demand equal to the sum of the original demands. To ensure that the flow goes to the right destinations, a super-destination node is created and connected to each original destination with an arc of capacity equal to the original demand of the commodity. When solving the \emph{linear relaxation} with grouped commodities, the same solution is computed at the condition that all commodities of a group originate from the same node. Grouping commodities can greatly reduce the computing time of a linear solution computed with an LP solver when the number of different origins is much smaller than the number of commodities. In our test cases, inspired by practical telecommunication instances, commodities are emitted from a small number of different origins thus it is efficient to group the commodities by origin. However, solutions produced when commodities are grouped yield a little less information. It is necessary to compute exactly what paths each commodity uses. This can be done quickly in $O(|V|(|E| + |K|))$ operations with a flow decomposition algorithm \citep{ford1956network}.

\section{A variation of the heuristic with approximation guarantees}
\label{sec:proof}

In this section, we present the Constrained Sequential Randomized Rounding algorithm (CSRR). It is a variation of the SRR heuristic for which we prove approximation guarantees similar to the one of the RR algorithm of \citet{raghavan1987randomized}. Approximation results are obtained considering the classical congestion objective function. Indeed, for the overflow sum objective, the value of the optimal solution may be zero. This happens when all commodities fit in the capacities. In this case, any approximation algorithm must find the optimal solution. Thus the overflow sum objective function is not suited for approximation proofs.

The RR algorithm is able to yield a solution satisfying a $O\left(\frac{ \ln(|E|\epsilon^{-1})}{\ln( \ln(|E|\epsilon^{-1}))}\right)$-approximation factor with probability $1-\epsilon$. We extend and tighten the analysis of randomized rounding algorithms by giving a new $O\left(\frac{\gamma \ln(|E|\epsilon^{-1})}{\ln(\gamma \ln(|E|\epsilon^{-1}))}\right)$ factor for the CSRR algorithm and for the RR algorithm of \citet{raghavan1987randomized}. In this new factor, $\gamma$ is the granularity parameter of the instance which is equal to $\frac{D_{max}}{c_{min} \Delta^*}$ where $\Delta^*$ is the optimal congestion of the linear relaxation. This parameter is small when the commodities are smaller compared to $c_{min} \Delta^*$. The parameter $\gamma$ can be related to the parameter $\frac{D_{max}}{c_{min}}$ introduced by \citet{chakrabarti2007approximation, shepherd2015inapproximability, azar2006combinatorial} to tighten their approximation analysis in the case of the maximum demand objective. The number $\gamma$ is a decisive parameter of unsplittable flow instances. Indeed, it remains constant when the capacities or the demands are scaled uniformly.

To prove an approximation factor for a randomized rounding where the linear relaxation is actualized (as in the line \ref{algo:lineLP} of Algorithm \ref{heuristic}), it appeared necessary to add a constraint to the linear relaxation. Thus the CSRR algorithm is the same as the SRR algorithm with the following additional constraint in the linear relaxation:
\begin{align}
    \sum_{k \in K \setminus K_{fixed}} f_{ek} D_k \leq c_e \Delta^* - \sum_{k \in K_{fixed}} \hat{f}_{ek} D_k, ~ \forall e \in E, \label{added_constraint}
\end{align}
where $\Delta^*$ is the optimal congestion of the first linear relaxation; $K_{fixed}$ is the set of commodities that were fixed to their respective single paths before the current linear relaxation resolution; $f_{ek}$ are the variables used to optimize the flow of the unfixed variables. The commodities in $K_{fixed}$ sent their flow on several paths in the linear relaxation before they were fixed to one path; $\hat{f}_{ek}$ is the corresponding fractional flow on each arc for these commodities. Note that for the first computation of the linear relaxation, this constraint has no impact on the final solution and can be removed. Thus, this constraint disappears in the algorithm of \citet{raghavan1987randomized}, since there is only one resolution of the linear relaxation.

For the sake of clarity, we derive the present proof in the case where the CSRR algorithm makes only a single actualization of the linear relaxation's solution which occurs just after the first rounding step. Extension to the case of several actualizations performed at any time can be done by induction.

In the following, let the discrete random variables $F_{ek}$ indicate the flow of commodity $k$ on arc $e$ in the solution returned by the CSRR algorithm. The variables $F_{ek}$ take the value $D_k$ with probability $f_{ek} = \sum_{p \in P_k | e \in p} x_{pk}$ and 0 otherwise. Thus, their expectation $\mathbb{E}[F_{ek}] = f_{ek} D_k = D_k \sum_{p \in P_k | e \in p} x_{pk}$ is also the flow of commodity $k$ on arc $e$ in the solution of the linear relaxation. Let $k_1$ be the index of the commodity of largest demand (thus the first one to be fixed in Algorithm \ref{heuristic}) and let $F_e = \sum_{k \in K \setminus \{k_1\}} F_{ek}$. Once conditioned by $F_{ek_1}$, the random variables $F_{ek}$ ($k \neq k_1$) are independent of each other because the linear relaxation is not actualized between their rounding step. However, the random variables $F_{ek}$ ($k \neq k_1$) are not independent of $F_{ek_1}$; in particular, we have $\mathbb{E}[F_{ek} | F_{ek_1}] \neq \mathbb{E}[F_{ek}]$. Indeed the realization of $F_{ek_1}$ in the unique rounding step conditions the resolution of the subsequent linear relaxation. Thus, it conditions the values $f_{ek}$ which parametrize the distribution of the random variables $F_{ek}$. To recall this dependency, we write $f_{ek}(F_{ek_1})$ the fractional flow of commodity $k$ on arc $e$. Note that constraint \eqref{added_constraint} added in the CSRR algorithm can be re-written in terms of random variables:
$$\mathbb{E}[F_e | F_{ek_1}] \leq c_e \Delta^* - \mathbb{E}[F_{ek_1}].$$
Recall that in the congestion formulation, the objective function aims at minimizing the minimum multiplicative factor on all arc capacities needed to fit the commodities. We note $C^*$ the optimal congestion for the considered unsplittable flow instance. As introduced above, $F_{ek_1} + F_e$ is the flow on arc $e$ in the solution returned by the CSRR algorithm. Thus, proving a probabilistic $(1+\alpha)$-approximation for this algorithm bound boils down to proving that for all arcs, $F_{ek_1} + F_e$ remains below $(1+\alpha) c_e C^*$ with high probability. Formally, for a small $\epsilon$:
$$\mathbb{P} \left( \forall e\in E, F_{ek_1} + F_e \leq (1+\alpha) c_e C^* \right) \geq 1-\epsilon.$$
Conversely, this is equivalent to proving that, with at most probability $\epsilon$, there exists an arc $e$ where the congestion exceeds $(1+\alpha) C^*$:
$$\mathbb{P} \left( \exists e\in E, F_{ek_1} + F_e \geq (1+\alpha) c_e C^* \right) \leq \epsilon.$$
To that end, we prove in Theorem \ref{theorem1} that for every arc $e$:
$$\mathbb{P} \left(F_{ek_1} + F_e \geq (1+\alpha) c_e \Delta^* \right) \leq \frac{\epsilon}{|E|}.$$
Indeed, in this case, as $\Delta^*$ is a lower bound on $C^*$, we have:
\begin{align*}
    \mathbb{P} \left( \exists e\in E, F_{ek_1} + F_e \geq (1+\alpha) c_e C^* \right) & = \mathbb{P} \left( \bigvee_{e\in E} F_{ek_1} + F_e \geq (1+\alpha) c_e C^* \right)\\
    & \leq \mathbb{P} \left( \bigvee_{e\in E} F_{ek_1} + F_e \geq (1+\alpha) c_e \Delta^* \right) \\
    & \leq \sum_{e\in E} \mathbb{P} \left( F_{ek_1} + F_e \geq (1+\alpha) c_e \Delta^* \right) \\
    & \leq \sum_{e\in E} \frac{\epsilon}{|E|}\\
    & = \epsilon
\end{align*}
To ensure that $\mathbb{P} \left(F_{ek_1} + F_e \geq (1+\alpha) c_e \Delta^* \right) \leq \frac{\epsilon}{|E|}$, we first prove through Lemma \ref{lem:lemma2} that the probability $\mathbb{P} \left(F_{ek_1} + F_e \geq (1+\alpha) c_e \Delta^* \right)$ is upper bounded by a quantity $g_e(\alpha)$. Lemma \ref{lem:lemma2} is proved by using the Markov inequality and the probabilistic translation of constraint \eqref{added_constraint}. Proving Lemma \ref{lem:lemma2} requires a preliminary result introduced in Lemma \ref{lem:lemma1}. Finally, the proof is completed by showing that there exists a value for $\alpha$  satisfying $1+\alpha = O\left(\frac{\gamma \ln(|E|\epsilon^{-1})}{\ln(\gamma \ln(|E|\epsilon^{-1}))}\right)$ and for every arc $g_e(\alpha) \leq \frac{\epsilon}{|E|}$

Without loss of generality and to remove $D_{max}$ from the proof, we assume the considered instances are scaled so that $D_{max} = 1$ and thus $\gamma = (c_{min} \Delta^*)^{-1}$. We now present the two Lemmas together with their proof.

\begin{lemma}
\label{lem:lemma1}
For any positive scalar $\alpha$, $\mathbb{E}\left[ (1+\alpha)^{F_e} | F_{ek_1}\right] \leq e^{\alpha \mathbb{E}[F_e | F_{ek_1}]}$.
\end{lemma}

\begin{proof}

\begin{alignat*}{3}
& \mathbb{E}\left[(1+\alpha)^{F_e} | F_{ek_1}\right] ~ = ~ \mathbb{E}\left[ \prod_{k \in K \setminus \{k_1\}} (1+\alpha)^{F_{ek}} | F_{ek_1}\right]\\
& ~ = ~ \prod_{k \in K \setminus \{k_1\}} \mathbb{E}\left[(1+\alpha)^{F_{ek}} | F_{ek_1}\right] \quad \text{because the $F_{ek}|F_{ek_1}$ are independent} \\
& ~ = ~ \prod_{k \in K \setminus \{k_1\}} (f_{ek}(F_{ek_1})(1+\alpha)^{D_k} + 1 - f_{ek}(F_{ek_1})) & \quad &\\
& ~ \leq ~ \prod_{k \in K \setminus \{k_1\}} (f_{ek}(F_{ek_1})(1+\alpha D_k) + 1 - f_{ek}(F_{ek_1}))   \quad \text{because }  D_k \leq 1\\
& ~ = ~ \prod_{k \in K \setminus \{k_1\}} (1 + \alpha f_{ek}(F_{ek_1}) D_k) & \quad &\\
& ~ \leq ~ \prod_{k \in K \setminus \{k_1\}} e^{\alpha f_{ek}(F_{ek_1}) D_k} & \quad &\\
& ~ = ~  e^{\alpha \sum_{k \in K \setminus \{k_1\}} f_{ek}(F_{ek_1}) D_k} & \quad &\\
& ~ = ~ e^{\alpha \mathbb{E}[F_e | F_{ek_1}]} & \quad &\\
\end{alignat*}
\qed
\end{proof}

\begin{lemma}
\label{lem:lemma2}
For any positive scalar $\alpha$ and any instance of the unsplittable flow problem, the flow $F_{ek_1} + F_e$ returned by the CSRR algorithm on arc $e$ satisfies:
$$\mathbb{P}(F_{ek_1} + F_e \geq (1+\alpha)c_e\Delta^*) \leq \left[\frac{e^\alpha}{(1+\alpha)^{1 + \alpha}}\right]^{c_e\Delta^*}$$
\end{lemma}

\begin{proof}

\quad We note $\delta = (1+\alpha)^{(1 + \alpha)c_e\Delta^*}$

\begin{alignat*}{3}
&\mathbb{P}(F_{ek_1} + F_e \geq (1+\alpha)c_e\Delta^*) \\
& ~ = ~ \mathbb{P}\left((1+\alpha)^{F_{ek_1} + F_e} \geq \delta \right)  & \quad &\\
& ~ \leq ~ \delta^{-1} \mathbb{E}\left[(1+\alpha)^{F_{ek_1} + F_e}\right] & \quad & \text{Markov inequality}\\
& ~ = ~ \delta^{-1} \mathbb{E}\left[\mathbb{E}[(1+\alpha)^{F_{ek_1} + F_e}|F_{ek_1}]\right] & \quad &\\
& ~ = ~ \delta^{-1} \mathbb{E}\left[(1+\alpha)^{F_{ek_1}}\mathbb{E}[(1+\alpha)^{F_e}|F_{ek_1}]\right] & \quad &\\
& ~ \leq ~ \delta^{-1} \mathbb{E}\left[(1+\alpha)^{F_{ek_1}}e^{\alpha \mathbb{E}[F_e | F_{ek_1}]}\right] & \quad & \text{Lemma 1}\\
& ~ \leq ~ \delta^{-1} \mathbb{E}\left[(1+\alpha)^{F_{ek_1}}e^{\alpha (c_e \Delta^* - \mathbb{E}[F_{ek_1}])}\right] & \quad & \text{because of constraint \eqref{added_constraint}}\\
& ~ = ~ \delta^{-1} e^{\alpha (c_e \Delta^* - \mathbb{E}[F_{ek_1}])} \mathbb{E}\left[(1+\alpha)^{F_{ek_1}}\right] & \quad &\\
& ~ = ~ \delta^{-1} e^{\alpha (c_e \Delta^* - \mathbb{E}[F_{ek_1}])} (f_{ek_1} (1+\alpha)^{D_{k_1}} + 1 - f_{ek_1}) \\
& ~ \leq ~ \delta^{-1} e^{\alpha (c_e \Delta^* - \mathbb{E}[F_{ek_1}])} (f_{ek_1} (1+\alpha D_{k_1}) + 1 - f_{ek_1}) & \quad & \text{because } D_{k_1} \leq D_{max} \leq 1\\
& ~ \leq ~ \delta^{-1} e^{\alpha (c_e \Delta^* - \mathbb{E}[F_{ek_1}])} e^{\alpha f_{ek_1} D_{k_1}} & \quad &\\
& ~ = ~ \delta^{-1} e^{\alpha (c_e \Delta^* - \mathbb{E}[F_{ek_1}] + \mathbb{E}[F_{ek_1}])} & \quad & \\
& ~ = ~ \delta^{-1} e^{\alpha c_e \Delta^*} & \quad & \\
& ~ = ~ \left[\frac{e^\alpha}{(1+\alpha)^{1 + \alpha}}\right]^{c_e\Delta^*} & \quad &
\end{alignat*}
\qed
\end{proof}

We now present the main theorem which upper-bounds for every arc the probability of high congestion in the solution returned by the CSRR algorithm.

\begin{theorem}
\label{theorem1}
For any $\epsilon > 0$, there exists an approximation factor $1+\alpha$ which is $O\left(\frac{\gamma \ln(|E|\epsilon^{-1})}{\ln(\gamma \ln(|E|\epsilon^{-1}))}\right)$ such that, for any instance of the unsplittable flow problem, the flow $F_{ek_1} + F_e$ returned by the CSRR algorithm on arc $e$ satisfies:
$$\mathbb{P} \left(F_{ek_1} + F_e \geq (1+\alpha) c_e \Delta^* \right) \leq \frac{\epsilon}{|E|}.$$
\end{theorem}

\begin{proof}

Lemma \ref{lem:lemma2} gives us $\mathbb{P}(F_{ek_1} + F_e \geq (1+\alpha)c_e\Delta^*) \leq \left[\frac{e^\alpha}{(1+\alpha)^{1 + \alpha}}\right]^{c_e\Delta^*}$. To ensure the veracity of the theorem we need to show that there exists a scalar $1+\alpha$ which is $O\left(\frac{\gamma \ln(|E|\epsilon^{-1})}{\ln(\gamma \ln(|E|\epsilon^{-1}))}\right)$ and for every arc satisfies:

$$\left[\frac{e^\alpha}{(1+\alpha)^{1 + \alpha}}\right]^{c_e\Delta^*} \leq \frac{\epsilon}{|E|}$$

$$\Longleftrightarrow$$
$$(1+\alpha)\ln(1+\alpha) - \alpha \geq \frac{\ln(|E|\epsilon^{-1})}{c_e\Delta^*}$$

For the arc of capacity $c_{min}$ which gives the highest bound, the lower bound is $B = \gamma \ln(|E|\epsilon^{-1})$ (recall that $D_{max} = 1$). Thus we study the solution of the equation $(1+\alpha)\ln(1+\alpha) - \alpha = B$ and show that it satisfies $1+\alpha = O \left(\frac{B}{\ln(B)} \right)$. By replacing $\ln(1+x)$ with the classical bounds $2\frac{x-1}{x+1} \leq \ln(1+x) \leq x$, we have:
\begin{align*}
& \alpha^2 \geq (1+\alpha)\ln(1+\alpha) - \alpha  \geq \alpha - 2\\
& \Longleftrightarrow \quad \alpha^2 \geq B \geq \alpha - 2\\
& \Longleftrightarrow \quad \sqrt{B} \leq \alpha \leq B + 2
\end{align*}
Which implies:
$$B =  (1+\alpha) \ln(1+\alpha) - \alpha \geq (1+\alpha) \ln(1+\sqrt{B}) - B - 2$$
and finally,
$$1+\alpha \leq \frac{2B+2}{\ln(1+\sqrt{B})} \sim \frac{4B}{\ln(B)} = O\left(\frac{B}{\ln(B)}\right)$$

Thus, the solution of the equation $(1+\alpha)\ln(1+\alpha) - \alpha = \gamma \ln(|E|\epsilon^{-1})$ satisfies $1+\alpha = O\left(\frac{\gamma \ln(|E|\epsilon^{-1})}{\ln(\gamma \ln(|E|\epsilon^{-1}))}\right)$ and for this value of $\alpha$, we have 
$$\mathbb{P} \left(F_{ek_1} + F_e \geq (1+\alpha) c_e \Delta^* \right) \leq \frac{\epsilon}{|E|}$$.
\qed
\end{proof}

Using Theorem \ref{theorem1}, we showed that, for any instance of the unsplittable flow problem, the CSRR algorithm returns a solution whose congestion is less than $O\left(\frac{\gamma \ln(|E|\epsilon^{-1})}{\ln(\gamma \ln(|E|\epsilon^{-1}))} \Delta^* \right)$ with probability $1-\epsilon$. As for the RR algorithm of \citet{raghavan1987randomized}, Lemma \ref{lem:lemma2} and Theorem \ref{theorem1} still apply even though the demonstration of Lemma \ref{lem:lemma2} is simpler because commodity $k_1$ does not need to be treated separately. We thus have the same approximation results for the RR algorithm.

During the resolutions of the linear relaxation of the CSRR algorithm, the constraint \eqref{added_constraint} is very restrictive. It almost does not leave any flexibility to the variables to move away from the previous optimum. 

\textit{Example: We call here first linear relaxation the linear program solved before any commodity is fixed to a single path and second linear relaxation the linear program solved after the first set of commodities is fixed to a single path. For the first and second linear relaxation respectively, we note $f_e^1$ and $f_e^2$ the total flow of all the commodities that are not fixed in the first rounding step. Let us suppose that the first linear relaxation has a unique optimal solution in which every arc has the same congestion $\Delta^*$. Then, in the second linear relaxation, constraint \eqref{added_constraint} ensures that, $\forall e \in E, f_e^1 \geq f_e^2$. In this case, we must have $\forall e \in E, f_e^1 = f_e^2$. Otherwise, the flow of the unfixed commodities in the first linear relaxation could be replaced by the same flow from the second linear relaxation thus creating a new optimal solution of the first linear relaxation. This would contradict the assumption that the first linear relaxation has a unique optimal solution. Thus, in this example, the variables of the second linear relaxation cannot move away from the previous optimum at all.}

This example highlights that the second linear relaxation can only move away from the solution of the first linear relaxation by using the slack left between the congestion of each arc and $\Delta^*$. Thus, when constraint \eqref{added_constraint} is used, the actualization step only has a small impact on the linear relaxation's solution and the CSRR algorithm does not plainly benefit from the actualization step. We conjecture that it is why the CSRR algorithm does not yield experimentally good results compared to the SRR heuristic (see experimental results of the CSRR algorithm in Section \ref{sec:results}). To overcome this problem, we replace constraint \eqref{added_constraint} by the following constraint: 
\begin{alignat}{1}
\label{new_added_constraint}
\sum_{k \in K_{free}} f_{ek} D_k \leq \beta c_e \Delta^* - \sum_{k \in K_{fixed}} \hat{f}_{ek} D_k, ~ \forall e \in E.
\end{alignat}
By replacing $\Delta^*$ by $\beta \Delta^*$ and $\gamma$ by $\beta^{-1} \gamma$ for some $\beta \geq 1$ in the proof, one can prove that using constraint \eqref{new_added_constraint} the approximation factor of the CSRR algorithm becomes $O\left(\frac{\gamma \ln(|E|\epsilon^{-1})}{\ln(\beta^{-1} \gamma \ln(|E|\epsilon^{-1}))}\right)$ which is still equal to $O\left(\frac{\gamma \ln(|E|\epsilon^{-1})}{\ln(\gamma \ln(|E|\epsilon^{-1}))}\right)$ for any fixed $\beta$. Unlike constraint \eqref{added_constraint}, in most practical cases, constraint \eqref{new_added_constraint} is not active in the optimal solutions of the linear relaxations even for $\beta = 1.1$. Thus, constraint \eqref{new_added_constraint} has no impact on the  practical computations and enables the CSRR algorithm to yield the same experimental results as the SRR heuristic.

To summarize, in this section, we added a constraint to the SRR heuristic and switched back to the congestion formulation to create an algorithm that has the same approximation guarantees as the RR algorithm of \citet{raghavan1987randomized}. We also tightened the approximation analysis of both algorithms by introducing a granularity parameter $\gamma$. This parameter has the property of remaining constant when commodities and arc capacities are uniformly scaled. Finally, we slightly modified the added constraint to alleviate its impact in practice. This modification increased the approximation factor by a negligible value.

\section{Experimental results}
\label{sec:results}

In this section, we present experiments that support our claims: the SRR heuristic has a lower computing time on large instances than exact methods and yields solutions with a better overflow than the algorithm of \citet{raghavan1987randomized} and meta-heuristics. The impact of sorting the commodities in decreasing order of demand in the randomized rounding algorithms is also investigated. Moreover, the SRR heuristic is compared to the CSRR approximation algorithm. The datasets and the code used in the experimental section of this work are accessible at \url{https://github.com/SuReLI/randomized_rounding_paper_code}. All the code for this work was written in Python 3. The experiments were made on a server with 48 CPU Intel Xeon E5-2670 2.30GHz, 60 Gbit of RAM, and CentOS Linux 7.

In this section, each figure presents the results for a dataset of instances. Each dataset features ten groups of a hundred instances with each group containing instances created using the same parameters. An exception is made in Figure \ref{size_scaling_MILP} in which each group contains only ten instances. Each point in a figure reports, for an algorithm, the average result of one group of instances. The $95 \%$ confidence intervals are also represented as semi-transparent boxes around the main curve.

\subsection{Instance datasets}

As stated in \citep{masri2019metaheuristics}, no standard benchmark of instances is present in the literature for the unsplittable flow problem, especially for large instances. Indeed, most works use small graphs (less than 50 nodes) on which they manually generate a set of commodities. The largest instances (100 nodes) can be found in the work of \citet{masri2015multi} and \citet{li2010ant}. \citet{masri2015multi} use a grid graph while  \citet{li2010ant} create their graphs with an adaptation of the graph generator NETGEN. To compensate for this absence of benchmark, we give a detailed explanation of our instance generation procedure. Moreover, all our instances are given in our GitHub repository together with the code used for their generation.

In our tests, we consider two types of graphs: strongly connected random graphs and grid graphs. For strongly connected random graphs, we use the following method to construct a random very sparse strongly connected graph: select a random node $u$, select a random $v$ such that there is no path from $u$ to $v$, add an arc $(u,v)$ to the graph, repeat until the graph is strongly connected. Afterward, random edges are added to control the average degree of the graph (it cannot be less than 2). In our tests, the average degree is fixed to $5$ and the probability of a node being an origin is $1/10$. A grid graph is an $n \times m$ toric grid with $p$ additional nodes. Each additional node is randomly connected to $q$ nodes on the grid and serves as an origin of the flow. In our tests, we use $m = n =  p = \frac{q}{2}$. Unless mentioned otherwise, the arc capacities are $10^4$.

For both types of graphs the demand is created as follows until no more commodity can be added without breaking the capacity constraints:
\begin{itemize}
    \item choose a destination node $d$;
    \item choose an origin $o$ which can access $d$ within the remaining capacities;
    \item compute a random simple path $p$ from $o$ to $d$ using a depth-first search where the visit order of newly discovered nodes is random;
    \item choose a demand level $D$ uniformly between 1 and $\hat{D}_{\max}$ where $\hat{D}_{\max}$ is a parameter defining the maximum possible demand of a commodity; if the chosen demand level is larger than the remaining capacity on $p$ then truncate it to the remaining capacity;
    \item decrease the capacities on the path $p$ by $D$;
    \item add $(o, d, D)$ to the list of created commodities;
\end{itemize}

Note that demands created this way can always be routed within the arc capacities and the optimal congestion is one. Thus, we know optimal solutions have zero overflow and a congestion of one. Hence, these optimal values do not need to be computed with exact optimization methods. Moreover, the parameter $\hat{D}_{\max}$ used to parametrize the size of the commodities and thus the number of commodities. Unless mentioned otherwise the value of $\hat{D}_{\max}$ is fixed to 1500 in our tests.

\subsection{Benchmarking meta-heuristics}
\label{sec:simulated_annealing}

In this section, we present a benchmark of ACO-MC, one of the ant colony optimization algorithms presented in \citep{li2010ant}, and the variable neighborhood search of \citet{masri2015multi}. Both algorithms were reproduced and the code used is given in our GitHub repository. We compare these algorithms with a handcrafted simulated annealing. The goal is to use only the best one as a comparison for the SRR algorithm in the next sections. The algorithm of \citet{li2010ant} was coded exactly as presented in their paper.

\textbf{Implementation of \citep{masri2015multi}:} due to differences in the considered unsplittable flow problem, the local search part of their algorithm had to be modified. Their local search creates a new path for a commodity by starting from its destination and adding arcs to the path until the origin is reached. At each step, the next arc is chosen by considering the following heuristic information for each out-going arc of the last node of the current path:
$$I_e = \frac{1}{l_e} + \left (1 - \frac{1}{\hat{c}_e} \right)$$
where $l_e$ is the lead time of arc $e$ (\emph{i.e.} its length) and $\hat{c}_e$ is the remaining capacity on arc $e$. As we do not have lead times for each arc in our problem, $l_e$ was set to 1 for each arc. Moreover, in the problem studied in \citep{masri2015multi}, the remaining capacity $\hat{c}_e$ is positive because no overflow is allowed. Because we can have negative remaining capacities, we replace the function $f : x \rightarrow 1 - \frac{1}{x}$ applied to $\hat{c}_e$ by $g : x \rightarrow \frac{1}{2} \left(1 + \frac{x}{1 + |x|}\right)$. Function $g$ was chosen to have similar properties to the function $f$.
$$\forall x \in ~ (1, +\infty), ~ 0 < f(x) < 1$$ 
$$\forall x \in ~ (-\infty, +\infty), ~ 0 < g(x) < 1$$
$$f(x) - 1 \underset{+\infty}{\sim} g(x) - 1 \underset{+\infty}{\sim} \frac{-1}{x}$$
$$g(x) \underset{-\infty}{\sim} \frac{-1}{x}$$
In our tests, we also present the results obtained with a version of the algorithm of \citet{masri2015multi} where the local search part is disabled.

\textbf{Our simulated annealing:} at each iteration, a solution is created in the neighborhood of the current solution; this modification is accepted with a probability depending on the improvement/deterioration of the solution and a temperature parameter. At each iteration, the temperature parameter is multiplied by $1 - \epsilon$ for some small $\epsilon$ depending on the number of iterations. At the beginning of the simulated annealing procedure and similarly to the algorithm of \citet{masri2015multi}, a list of $k$-shortest paths of length 10 is computed for each commodity using the algorithm of \citet{jimenez1999computing}. To initialize the solution, each commodity takes a random path in its list of $k$-shortest paths. At each iteration, a neighborhood solution is created by replacing the path of a commodity with a path randomly chosen in the list of $k$-shortest paths of the commodity. The stopping criterion is the total number of iterations.

\textbf{Hyper-parameter setting:} we benchmarked different hyper-parameters values for ACO-MC but finally settled to use the same values as in \citep{li2010ant}. Except for the size of the largest neighborhood which is not mentioned, both versions of the variable neighborhood search use the hyper-parameter values given in \citep{masri2015multi}. After testing different values, the size of the largest neighborhood was set to 3.  The hyper-parameters of the simulated annealing are the initial and final temperature chosen to be respectively 200 and 1. The simulated annealing (SA) is given $2 |K|^{1.5}$ iterations so that it takes a time comparable to the SRR algorithm in the next sections. ACO-MC and the VNS of \citet{masri2015multi} were given respectively 50 and 100 iterations. Although these numbers seem very low, one must consider that for each of their iterations, the algorithms of \citet{li2010ant} and \citet{masri2015multi} respectively generate $|K|$ and $\frac{|K|}{2}$ new paths. With these numbers of iterations, they already require a longer computing time than the simulated annealing procedure. Finally, the variation of the variable neighborhood search where the local search part is disabled (VNS2) was given $|K|^{1.5}$ iterations.

\textbf{Results:} Figure \ref{size_scaling_metaheuristics} presents the results of each algorithm on a dataset composed of grid graphs of varying sizes. The simulated annealing procedure clearly outperforms the other algorithms on this dataset. We obtained similar results on all other datasets. That is why the simulated annealing algorithm has been chosen in the next sections to be the comparison point for the SRR algorithm.

\begin{figure}[ht]
     \centering
     \begin{subfigure}[b]{0.49 \columnwidth}
         \centering
         \includegraphics[width=\textwidth]{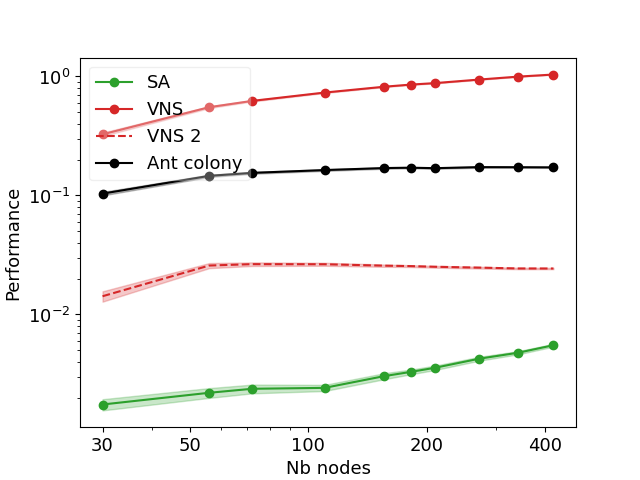}
         \caption{Overflow divided by the total demand --- grid graphs}
         \label{size_scaling_metaheuristics_res}
     \end{subfigure}
     \hfill
     \begin{subfigure}[b]{0.49 \columnwidth}
         \centering
         \includegraphics[width=\textwidth]{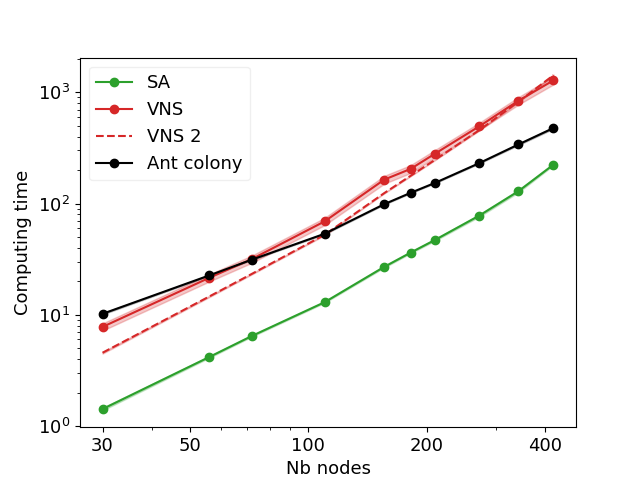}
         \caption{Computing time in seconds --- grid graphs}
         \label{size_scaling_metaheuristics_time}
     \end{subfigure}
        \caption{Performance and computing time versus the number of nodes of various meta-heuristics}
        \label{size_scaling_metaheuristics}
\end{figure}

\textbf{Discussion:} we now try to explain why the meta-heuristics from the literature are outperformed. At each iteration, the ant colony procedure spends most of its computing time creating an entirely new solution. Indeed, a new path is created for each commodity for a total of $|K|$ paths generated. Conversely, at each iteration, the simulated annealing procedure chooses only one path in a list (choosing a path is much faster than creating one). This leads to a much larger number of iterations for the simulated annealing procedure and thus a lot more solutions evaluated. As for the algorithm of \citet{masri2015multi}, it seems that the local search part performs poorly in our tests. This is partly because the path generation procedure does not know where its target node is until it is encountered. Indeed, the heuristic information used to choose the next arc of the path considers the length $l_e$ of the arcs and not the length of the shortest path to the target node. Thus, when choosing the next arc, the procedure does not know if it goes toward or away from its goal. We tried to replace $l_e$ by the length of the shortest path to the target node in the local search but disabling it completely still yielded better solutions. Finally, it appears that variable neighborhood search is not a good choice of meta-heuristic for our version of the unsplittable flow problem (especially when the number of commodities is more than a thousand). Indeed, by changing the size $k_{max}$ of the largest neighborhood considered we obtained that the best value for this parameter is $k_{max} = 1$ which is the case where only the smallest neighborhood is considered.

\subsection{Results}

We compare the following algorithms in our tests :

\begin{itemize}
    \item RR: the randomized rounding algorithm of \citet{raghavan1987randomized};
    \item RR sorted: a version of the RR algorithm where the commodities are rounded in order of decreasing demand;
    \item SRR: the Sequential Randomized Rounding heuristic described in Section \ref{sec:algo};
    \item SRR unsorted: a version of the SRR heuristic where the commodities are not rounded in order of decreasing demand but in random order;
    \item CSRR: the version of the SRR algorithm with approximation properties presented in Section \ref{sec:proof};
    \item SA and SA2: a handcrafted simulated annealing procedure presented in Section \ref{sec:simulated_annealing}; SA is given $2 |K|^{1.5}$ iterations to have a similar computing time to SRR on grid graphs while SA2 is given $6 |K|^{1.5}$ iterations;
    \item MILP solver: arc-node formulation solved with Gurobi 8.11.
\end{itemize}

All the randomized rounding algorithms in this list use the overflow-sum objective in the linear relaxation to create their solutions.

An important characteristic of the SRR heuristic is its capacity to scale to large instances. Thus we first illustrate how the algorithm outperforms a MILP method. In Figure \ref{size_scaling_MILP}, a MILP method based on the arc-node formulation is compared with the RR and the SRR algorithms on small grid graphs with very few commodities. The MILP algorithm solves optimally all the small instances but gives poor results on large instances within a 20 minutes time limit. In comparison, the other methods retain a reasonable performance on large instances. For these instances, the capacity of the arcs is 3 and the maximum demand of a commodity is 2 to keep the number of commodities and variables of the MILP formulation low. Due to the large amount of memory required, the MILP algorithm could not be tested on the other larger datasets.

\begin{figure}[ht]
     \centering
     \begin{subfigure}[b]{0.49 \columnwidth}
         \centering
         \includegraphics[width=\textwidth]{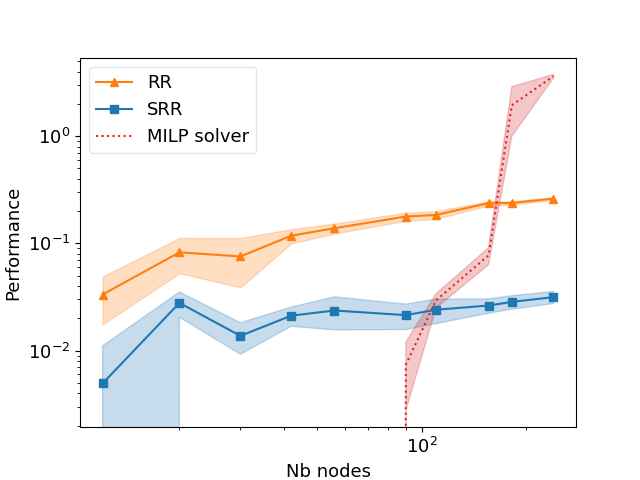}
         \caption{Overflow divided by the total demand --- grid graphs}
         \label{size_scaling_MILP_res}
     \end{subfigure}
     \hfill
     \begin{subfigure}[b]{0.49 \columnwidth}
         \centering
         \includegraphics[width=\textwidth]{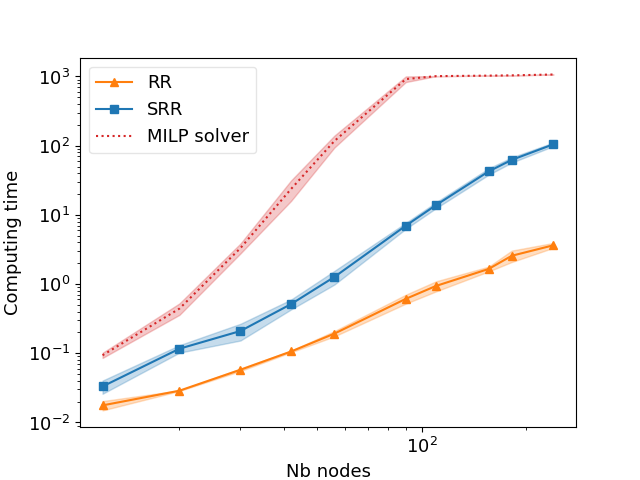}
         \caption{Computing time in seconds --- grid graphs}
         \label{size_scaling_MILP_time}
     \end{subfigure}
        \caption{Performance and computing time versus the number of nodes on small instances}
        \label{size_scaling_MILP}
\end{figure}

Figure \ref{size_scaling} compares different randomized rounding algorithms and the simulated annealing algorithm on grid graphs and random connected graphs. In both cases, the SRR algorithm requires a larger computing time than pure randomized rounding but returns solutions of much higher quality. For small instances of grid graphs, the best results are given by the simulated annealing procedure. However, as the number of nodes increases the SRR heuristics outperforms the simulated annealing procedure with the same computing time (SA). Furthermore, on the largest instances, the simulated annealing procedure is outperformed by the SRR heuristics even when given three times more computing time (SA2). As for strongly connected random graphs, with the described settings, the SRR heuristic returns the highest quality solutions but requires a larger computing time than the simulated annealing. It appears that solving the linear relaxation takes a longer time on random graphs than on grid graphs. The CSRR algorithm returns worse solutions than the SRR heuristic in a similar computing time. As explained at the end of Section \ref{sec:proof}, this is due to constraint \eqref{added_constraint} added in the resolution of the linear relaxation. One can relax this constraint into constraint \eqref{new_added_constraint} while keeping a similar approximation factor. In practice, constraint \eqref{new_added_constraint} is never active in the optimal solution of the linear relaxation. Thus, with this adaptation, the CSRR algorithm returns exactly the same solutions as the SRR heuristic.

\begin{figure}[ht]
     \centering
     \begin{subfigure}[b]{0.49 \columnwidth}
         \centering
         \includegraphics[width=\textwidth]{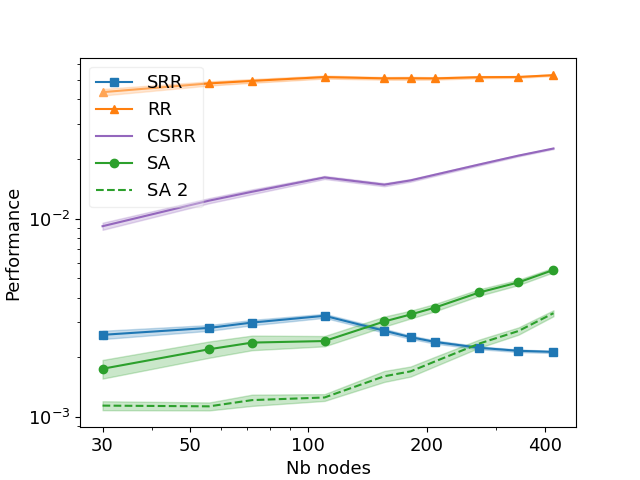}
         \caption{Overflow divided by the total demand --- grid graphs}
         \label{size_scaling_res}
     \end{subfigure}
     \hfill
     \begin{subfigure}[b]{0.49 \columnwidth}
         \centering
         \includegraphics[width=\textwidth]{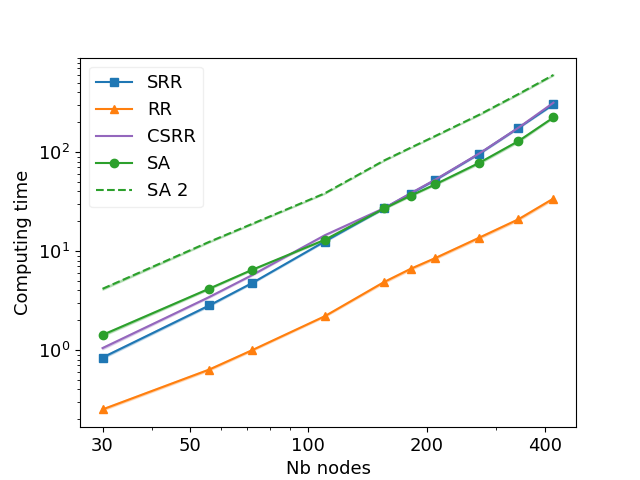}
         \caption{Computing time in seconds --- grid graphs}
         \label{size_scaling_time}
     \end{subfigure}
     \begin{subfigure}[b]{0.49 \columnwidth}
         \centering
         \includegraphics[width=\textwidth]{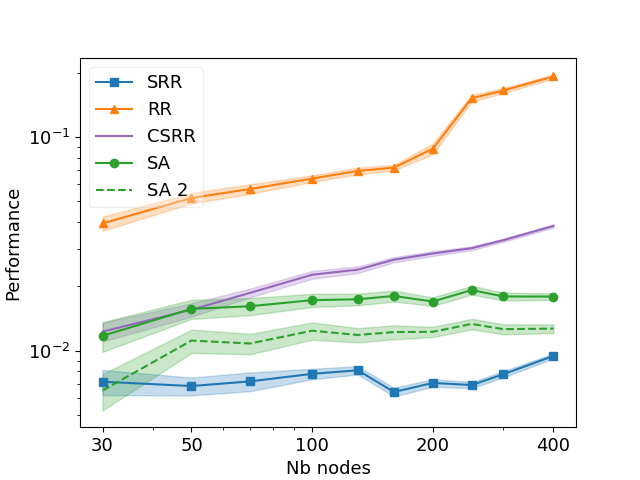}
         \caption{Overflow divided by the total demand --- random graphs}
         \label{size_scaling_random_res}
     \end{subfigure}
     \hfill
     \begin{subfigure}[b]{0.49 \columnwidth}
         \centering
         \includegraphics[width=\textwidth]{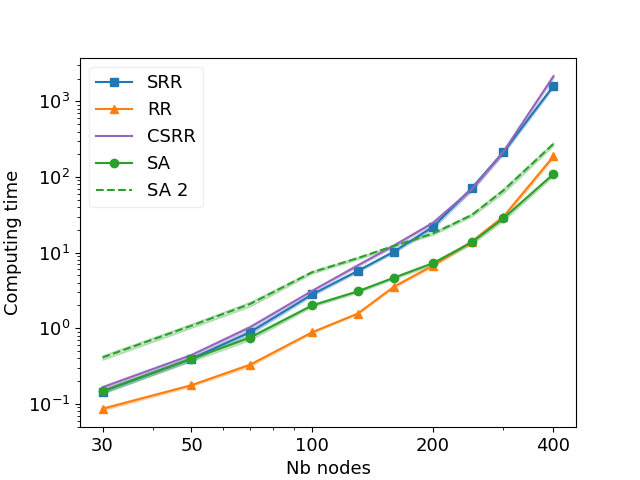}
         \caption{Computing time in seconds --- random graphs}
         \label{size_scaling_random_time}
     \end{subfigure}
        \caption{Performance and computing time versus the number of nodes on large instances}
        \label{size_scaling}
\end{figure}

Figure \ref{commodity_scaling} reports the solution performance and the computing time for grid graphs with a constant number of nodes and arcs but a various number of commodities. This variation is obtained by changing the ratio of the arc capacities over the maximum demand of a commodity. Table \ref{setting_commodity_scaling} reports the parameters used to create the instances. As can be seen in Figure \ref{commodity_scaling_res}, all algorithms produce higher quality solutions when there is a large number of commodities (\emph{i.e.} commodities have small demand compared to the arc capacities). As explained in Section \ref{sec:proof}, for randomized rounding algorithms, this practical result can be related to the presence of the granularity parameter $\gamma$ in the approximation factor of the RR and CSRR algorithms. Figure \ref{commodity_scaling_time} shows that our use of an aggregated arc-node formulation enables the randomized rounding algorithms to have a computing time that scales very well with the number of commodities. However, this is not the case for simulated annealing which requires a large increase in its number of iterations to obtain the solution reported in Figure \ref{commodity_scaling_res}.

\begin{figure}[ht]
     \centering
     \begin{subfigure}[b]{0.49 \columnwidth}
         \centering
         \includegraphics[width=\textwidth]{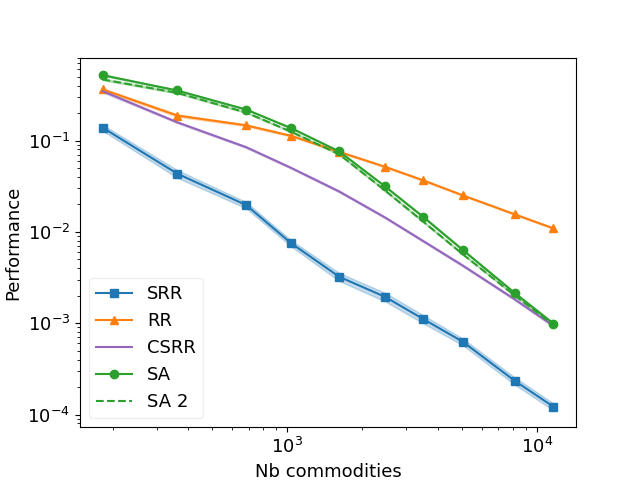}
         \caption{Overflow divided by the total demand --- grid graphs}
         \label{commodity_scaling_res}
     \end{subfigure}
     \hfill
     \begin{subfigure}[b]{0.49 \columnwidth}
         \centering
         \includegraphics[width=\textwidth]{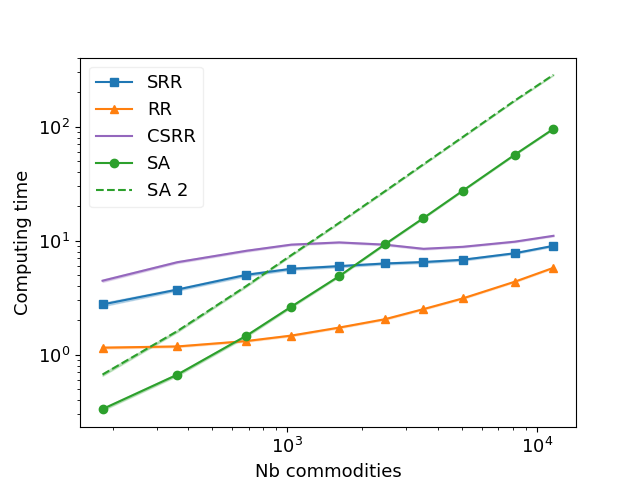}
         \caption{Computing time in seconds --- grid graphs}
         \label{commodity_scaling_time}
     \end{subfigure}
        \caption{Performance and computing time versus the number of commodities}
        \label{commodity_scaling}
\end{figure}

We now analyze the impact of rounding the commodities in order of decreasing demand, based on the results presented in Figure \ref{size_scaling_order}. Both pure randomized rounding and the SRR heuristic yield solutions of higher quality when the commodities are sorted. The impact is significantly higher on the SRR algorithm than on pure randomized rounding. A possible explanation for such a strong impact is given in Section \ref{sec:synergy}. At first glance, a rounding order should have no impact on pure randomized rounding since the commodities are rounded independently. However, results clearly indicate otherwise and this is due to how the linear relaxation is computed in our experiments. Indeed when using the arc-node formulation of Section \ref{sec:arc-node-formulation} and the commodity aggregation presented in Section \ref{sec:aggregation}, the linear program does not directly yield a flow distribution for each commodity. The flow distribution returned is for the aggregated commodities. To get the flow of each commodity, a flow decomposition algorithm is used. This algorithm tends to split less the commodities that are decomposed first. Thus when the commodities with the largest demands are rounded first, they are also decomposed first. This means the biggest commodities have a lower chance of being split and thus a lower chance that their rounding creates a large overflow.

\begin{figure}[ht]
     \centering
     \begin{subfigure}[b]{0.49 \columnwidth}
         \centering
         \includegraphics[width=\textwidth]{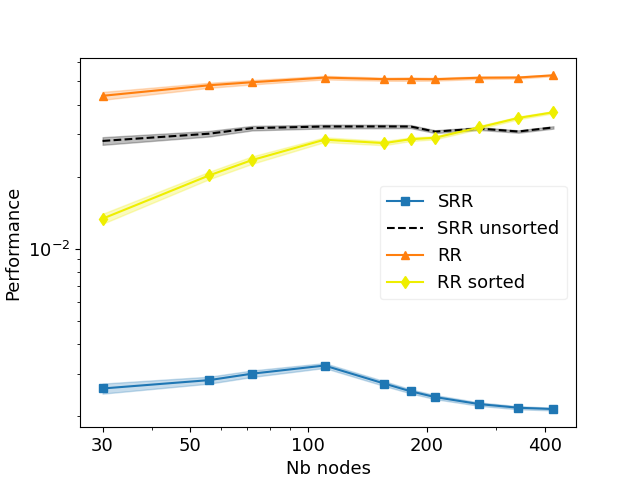}
         \caption{Overflow divided by the total demand --- grid graphs}
         \label{size_scaling_order_res}
     \end{subfigure}
     \hfill
     \begin{subfigure}[b]{0.49 \columnwidth}
         \centering
         \includegraphics[width=\textwidth]{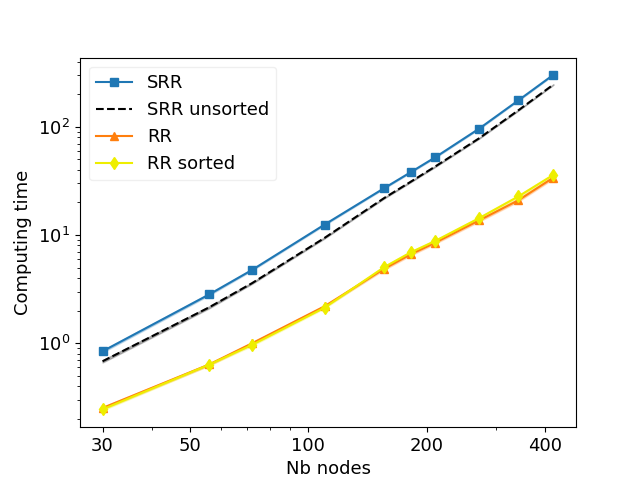}
         \caption{Computing time in seconds --- grid graphs}
         \label{size_scaling_order_time}
     \end{subfigure}
        \caption{Influence of commodity sorting on performance and computing time}
        \label{size_scaling_order}
\end{figure}

Finally, we describe how we chose the value $\frac{|V|}{4}$ for the actualization parameter $\theta$ of the SRR algorithm presented in Section \ref{sec:algo}. Setting this value implies making a trade-off between solution quality and computing time. Figure \ref{threshold_scaling} compares the performance of SRR for different values of this parameter, on a grid graph with 110 nodes. The computing time of the SRR heuristic affinely decreases with the $\theta$ parameter. However, the overflow of the returned solution decreases less and less when $\theta$ decreases. We decided to choose $\theta$ such that the returned solution is close to the best obtainable solution while a lot of computing time is saved compared to choosing $\theta$ close to zero. When repeating this experiment on graphs of different sizes and types, it appeared that the chosen trade-off value was close to $\frac{|V|}{4}$. This value is represented with a dashed line in Figure \ref{threshold_scaling}.

\begin{figure}[ht]
     \centering
     \begin{subfigure}[b]{0.49 \columnwidth}
         \centering
         \includegraphics[width=\textwidth]{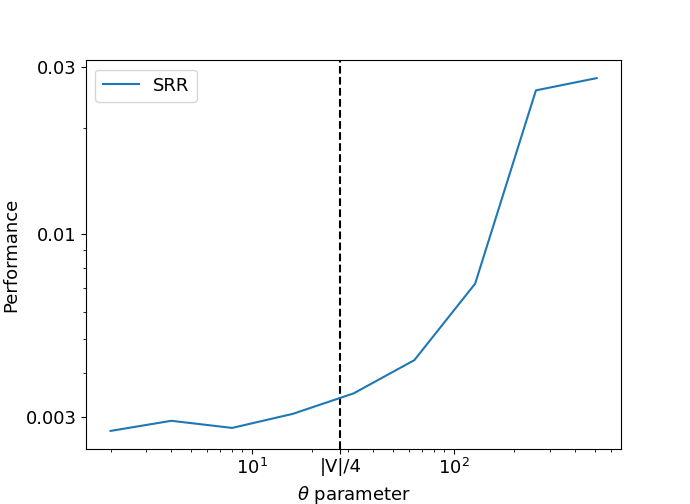}
         \caption{Overflow divided by the total demand --- grid graphs}
         \label{threshold_scaling_res}
     \end{subfigure}
     \hfill
     \begin{subfigure}[b]{0.49 \columnwidth}
         \centering
         \includegraphics[width=\textwidth]{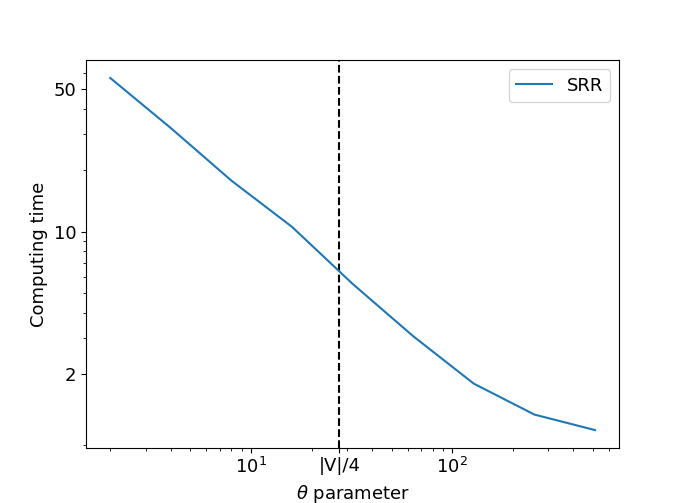}
         \caption{Computing time in seconds --- grid graphs}
         \label{threshold_scaling_time}
     \end{subfigure}
        \caption{Performance and computing time versus the value of the $\theta$ parameter of SRR}
        \label{threshold_scaling}
\end{figure}

\begin{table}
\centering
\begin{tabular}{|l|c|c|c|c|c|}
\hline
Arc capacities                & 1 & 2 & 5 & 10 & 20 \\ \hline
Maximum demand of a commodity & 1 & 1 & 2 & 3  & 4 \\ \hline
Average number of commodities & 182 & 362 & 685 & 1038 & 1615 \\ \hline
 \hline
Arc capacities & 50 & 100 & 200 & 500 & 1000 \\ \hline
Maximum demand of a commodity & 7  & 10  & 14  & 22  & 31   \\ \hline
Average number of commodities & 2462 & 3512 & 5048 & 8138 & 11644 \\ \hline
\end{tabular}
\caption{Parameters of the instances of the commodity scaling dataset : the graphs have 110 nodes and 580 arcs}
\label{setting_commodity_scaling}
\end{table}

\section{Discussion}
\label{sec:discussion}

In this section, we discuss different properties of the randomized rounding algorithms. We first show that using the overflow sum objective function presented in Section \ref{two_criteria} instead of the classical congestion objective function has a positive impact on practical results. Then we highlight a positive interplay between sorting the commodities and actualizing the linear relaxation.

\subsection{Impact of the objective function}

\label{compare_criteria}

Two types of objective function coexist in the SRR heuristic. The first one is the objective function of the unsplittable flow problem which will be called \emph{evaluation metric} in this section. The second one is the objective function used in the resolutions of the linear relaxation inside the SRR heuristic which we will call \emph{generation objective}. For both of the previous types, we study the two functions presented in Section \ref{two_criteria}: the sum of the overflow and the congestion. Usually, the generation objective and the evaluation metric are chosen to be the same. However, the overflow sum and the congestion are very close metrics and we want to study the impact of using one for creating solutions for the other. We will show below that using the overflow sum instead of the congestion as generation objective yields solutions of higher quality for both evaluation metrics. 

When the evaluation metric is the overflow sum, in our tests, solutions generated using the overflow sum as generation objective have a ten times better solution quality than the one created using the congestion. More surprisingly, the overflow sum used as generation objective also yields the best solutions when the evaluation metric is the congestion. To test this, the SRR heuristic has been applied with the two generation objectives on 100 grid graphs with 120 nodes, 700 arcs, and 4500 commodities. The mean congestion of the returned solutions is reported in the first two columns of Table \ref{congestion_comparison} together with the standard deviation of the results. The mean congestions for the two generation objectives are separated by more than one standard deviation in favor of the overflow sum. Using the values given in the first two columns of Table \ref{congestion_comparison}, an unpaired two-sample Student t-test was made. The test showed that the overflow sum used as generation objective performs significantly better than the congestion. The $p$-value associated with the test is $10^{-18}$.

\begin{table}[]
\centering

\begin{tabular}{|l|c|c|c|}
\hline
                   & Overflow sum & Congestion & Mixed \\ \hline
Mean               & 1.025        & 1.049      & 1.027 \\ \hline
Standard Deviation & 0.016        & 0.019      & 0.012 \\ \hline
\end{tabular}

\caption{Comparison of the congestion of the solutions returned by the SRR heuristic using different generation objectives.}

\label{congestion_comparison}
\end{table}

We conjecture that the difference observed above between the two generation objectives is explained by the following reasoning.  In the linear relaxation, the congestion objective function does not differentiate solutions having one arc with high congestion from solutions with several arcs with high congestion. However, a larger number of highly congested arcs in the linear relaxation implies, on average, a larger congestion in the solution returned by the randomized rounding process. Furthermore, if the randomized rounding process creates arcs with a larger congestion than the optimal congestion of the linear relaxation, then, in subsequent actualizations of the linear relaxation, the capacity constraints are lifted for all the other arcs. This might lead to even more congested unsplittable solutions.

With this conjecture in mind, we created an algorithm that uses the congestion as its main generation objective but performs as well as the algorithm using the overflow sum as generation objective. To that end, a two-level objective is set in the linear relaxation: find the solution of least overflow sum among the solutions with minimal congestion. The results obtained with this alternative generation objective are presented in the third column of Table \ref{congestion_comparison} under the name ``Mixed'' objective. These results are comparable to those obtained with the overflow sum objective. Indeed, an unpaired two-sample Student t-test made on the values given in the first and third columns of Table \ref{congestion_comparison} does not show that the difference is statistically significant. The $p$-value associated with the test is $0.32$.

\subsection{Interaction between sorting and actualization}
\label{sec:synergy}

In this section, we discuss the positive interplay observed in Section \ref{sec:results} between sorting the commodities by decreasing demand and actualizing the linear relaxation. To that end, we study a toy example where the graph is composed of two nodes linked by two parallels arcs. We also assume the sum of the demands is equal to the sum of the capacities. We compare the solutions with the overflow sum objective and assume that no binary solution has an overflow of zero. 

When computing the linear relaxation, optimal extrema of the polyhedron are of the following form: every commodity is assigned to one arc except for one commodity which is split between the two arcs. We assume that each commodity has the same probability of being the split commodity. After applying randomized rounding to a linear solution of this type, the binary solution obtained has an overflow between zero and the demand of the split commodity.

On the other hand, if we actualize the linear relaxation after fixing the split commodity we should often get a better solution. Indeed, after actualizing the linear solution, three situations can occur. Firstly, the remaining free commodities cannot switch arc to compensate any of the overflow generated by rounding the split commodity. In this case, no commodity changes its flow, the new linear solution is a binary solution and this solution is the same with and without actualization. Secondly, the remaining free commodities compensate only a part of the overflow created. In this case, even though some commodities change their flow, the new linear solution is a binary solution whose overflow is lower than the one before the actualization step. Lastly, the remaining free commodities compensate completely the overflow created and a new split commodity is created. Applying randomized rounding from there yields a solution whose overflow is between zero and the demand of the new split commodity. Since the commodities are fixed in order of decreasing demand, the new split commodity has a smaller demand than the old split commodity. The range in which the overflow of the final solution can vary is thus smaller which should on average lead to a better solution. Furthermore, new split commodities are created until the free commodities cannot completely compensate for the rounding of the last split commodity. Most of the time, this happens when only a few commodities remain and thus when only small commodities remain. In this case, the last split commodity has a small demand, and the algorithm yields a small final overflow.

Finally, we can see that in this toy example, actualizing the linear solution and sorting the commodities by decreasing demands yields, on average, solutions with a smaller overflow than pure randomized rounding. We conjecture that this phenomenon generalizes to any graph. Indeed, the fact that the split commodities get smaller and smaller as the algorithm progresses should lead to smaller and smaller final overflows.

\section{Conclusion}
\label{sec:conclusion}

In this paper, we have presented a heuristic based on randomized rounding for large-scale instances of the unsplittable flow problem which extends the algorithm of \citet{raghavan1987randomized}. We experimentally showed that on large-scale instances this heuristic produces solutions of smaller overflow than any other method used for comparison.

We also derived an approximation algorithm from the heuristic by restraining the possible actualization of the linear relaxation. The approximation factor of both this algorithm and the algorithm of \citet{raghavan1987randomized} was then tightened to $O\left(\frac{\gamma \ln(|E|\epsilon^{-1})}{\ln(\gamma \ln(|E|\epsilon^{-1}))}\right)$. This new approximation factor depends on the granularity parameter $\gamma$ and enables to understand the behavior of randomized rounding when the commodities are small compared to the capacities (\emph{i.e.}$\gamma \ll 1$).

Furthermore, the behavior of the presented heuristic has been analyzed to highlight two of its key particularities. First, the new objective function used in the algorithm for practical computations yielded solutions of higher quality. Secondly, the actualization of the solution of the linear relaxation enhanced the performances of the heuristic when the commodities are sorted in order of decreasing demand.

Finally, even though the techniques discussed in this paper were presented in the context of unsplittable flows, they apply to other contexts where the randomized rounding method of \citet{raghavan1987randomized} is used (packing problems, covering problems...). Their performances and variations in these contexts could be investigated. Moreover, the impact of backtracking of the decisions made during randomized rounding algorithms seems a promising research direction.


%

\bibliographystyle{spbasic}      

\bibliography{biblio.bib}

\end{document}